\documentclass[runningheads, envcountsame, a4paper]{article}

\usepackage[T1]{fontenc}
\usepackage{lmodern}
\usepackage{amssymb}
\usepackage{graphicx}
\usepackage{enumerate}
\usepackage{amsmath,amsopn,amsfonts}
\usepackage[hyphens]{url}
\usepackage{microtype}
\usepackage{amsthm}
\usepackage{tikz}
\usetikzlibrary{patterns}
\usepackage{comment}
\usepackage{hyperref}
\hypersetup{
	colorlinks=true,
	linkcolor=blue,
	filecolor=magenta,      
	urlcolor=cyan,
}

\usepackage{url}
\urldef{\mailsa}\path| berdinsky@gmail.com, 
Franks'email |

\date{}

\newtheorem{definition}{Definition}
\newtheorem{remark}[definition]{Remark} 
\newtheorem{theorem}[definition]{Theorem} 
\newtheorem{proposition}[definition]{Proposition} 
\newtheorem{corollary}[definition]{Corollary}
\newtheorem{lemma}[definition]{Lemma}
\definecolor{light-gray}{gray}{0.75}

\begin{document}
		
	\title{String Compression in FA--Presentable Structures}
	\author{Dmitry  Berdinsky
		\thanks{Mahidol University, Faculty of Science, Department of Mathematics and 
			Centre of Excellence in Mathematics, CHE,  
			Bangkok 10400, Thailand; e-mail: berdinsky@gmail.com} \and
		 Sanjay Jain
		 \thanks{Department of Computer Science, National University of Singapore, Singapore 117417, Republic
		 	of Singapore; e-mail: sanjay@comp.nus.edu.sg} 
		 \and  
		Bakhadyr Khoussainov
		\thanks{University of Electronic Science and Technology of China, Chengdu, P.R.China; e-mail: bmk@uestc.edu.cn}
	    \and  
		 Frank Stephan
		\thanks{Department of Mathematics, National University of Singapore, 10 Lower Kent Ridge Road, Singapore 119076, Republic of Singapore; e-mail: fstephan@comp.nus.edu.sg}
	    }
	\date{}
	\maketitle {\small
		\begin{quote}
			\noindent{\bf Abstract}\,\,\,\,   
		        
We construct a FA--presentation 
$\psi: L \rightarrow \mathbb{N}$ of the structure $(\mathbb{N};\mathrm{S})$ 
for which a numerical characteristic 
$r(n)$ defined as the maximum 
number $\psi(w)$ for all strings $w \in L$ of length less than or equal to $n$ 
grows faster than any tower of exponents 
of a fixed height. This result leads us 
to a more general notion of a compressibility 
rate defined for FA--presentations of any FA--presentable 
structure. We show the existence of 
FA--presentations for the configuration space 
of a Turing machine and Cayley graphs of 
some groups for which it grows faster than any 
tower of exponents of a fixed height. 
For FA--presentations of 
the Presburger arithmetic 
$(\mathbb{N};+)$ we show that it is bounded
from above
by a linear function.

		\end{quote}	

\section{Introduction}
   
   A \emph{FA--presentable 
   structure} 
   is a relational structure 
   $\mathcal{A}= (D; R_1,\dots,R_k)$ admitting presentations by finite automata. 
   In brief, for a FA--presentable structure 
   $\mathcal{A}$
   there should exist a 
   surjective map $\psi : L \rightarrow D$ between some regular language 
   $L$ and the domain $D$ of the structure $\mathcal{A}$ such that each relation $R_i$, $i=1,\dots,k$ is recognized by 
   a multi--tape synchronous finite automaton
   and the equality relation 
   $\{(u,v) \in L^2 \, : \, \psi(u) = \psi(v) \}$
   is recognized by a two--tape synchronous
   automaton. 
   The language $L$ can be thought of as a language of normal forms (not necessarily unique) for elements of $D$. 
   The map $\psi: L \rightarrow D$ is called a 
   FA--presentation of the structure $\mathcal{A}$. 
   
   FA--presentable structures are 
   often referred to as 
   \emph{automatic structures} in the literature. 
   The term automatic structure
   is also used in the theory of automatic groups 
   \cite{Epsteinbook}, but with the different meaning. 
   In order to avoid misinterpretation, in this paper we  use the 
   term FA--presentable structure.    
   The field  of FA--presentable structures  can be  
   traced back to the pioneering works by Hodgson 
   \cite{Hodgson82,Hodgson83}. The systematic study 
   of FA--presentable structures was initiated independently 
   by Khoussainov and Nerode \cite{KhoussainovNerode95} and 
   Blumensath and Gr\"{a}del \cite{BlumensathGradel00,BlumensathGradel04}.
   For survey articles in FA--presentable structures the reader is referred 
   to \cite{Gradel_survey_2020,khoussainov2007three,Rubin_survey08,Stephan2015}.

   Each FA--presentable structure admits infinitely many 
   FA--presentations which can differ from each other significantly 
   or may exhibit unexpected behaviour
   compared to natural FA--presentations. For example, in \cite{Akiyama08}  the 
   authors construct a FA--presentation of $(\mathbb{Z}[1/p];+)$ 
   for which the subgroup of integers $\mathbb{Z} \leqslant \mathbb{Z}[1/p]$
   is not regular and in \cite{NiesSemukhin07} 
   the authors construct FA--presentations of $(\mathbb{Z}^2;+)$ 
   for which none of the cyclic subgroups is regular.  
        
   In this paper we look at FA--presentations from 
   a numerical perspective.  We define a numerical characteristic of a FA--presentation $\psi: L \rightarrow \mathbb{N}$
   of the structure $(\mathbb{N};\mathrm{S})$ as follows. 
   Let $r(n)$ be the maximum $\psi(w)$ for all strings $w \in L$ of length less than or equal to $n$. 
   For example, for a unary presentation of 
   $(\mathbb{N};\mathrm{S})$ the function $r(n)$
   has a linear growth while  
   for a binary presentation it grows like
   an exponential function.              
    For infinitely many positive integers 
    $n$, those for which $r(n-1) < r(n)$,  the value $\frac{r(n)}{n}$ can be thought 
    of as a compression ratio -- for these integers $n$
    the number $r(n)$ is represented by a string of 
    length exactly $n$.    
    We first notice that for each FA--presentation of 
    the Presburger arithmetic $(\mathbb{N};+)$ the growth 
    of $r(n)$ is at most exponential, see Lemma
    \ref{N_plus_exp_growth}. 
    Then we show that in general it is not true for FA--presentations 
    of $(\mathbb{N};\mathrm{S})$ which comprise all 
    FA--presentations of $(\mathbb{N};+)$.     
    Namely, we construct a FA--presentation of 
    $(\mathbb{N};\mathrm{S})$ for which $r(n)$  
    grows at least as fast as the function $T(n)$       
   defined recursively by the identity 
   $T(n+1)=2^{T(n)}$ for $n \geqslant 0$
   and the initial condition 
   $T(0)=1$, see Theorem \ref{r_grows_faster_than_T}. In particular, $r(n)$ grows faster than any 
   tower of exponents of a fixed height
   (see Corollary 
   \ref{r_grows_faster_t_h_cor})\footnote{The original motivation for considering  
   such FA--presentations of 
   $(\mathbb{N};\mathrm{S})$ came from 
   the study of a so--called 
   Cayley distance function 
   \cite{BT18,BET21_IJAC}
   defined for FA--presentations 
   of Cayley graphs of Cayley automatic groups~\cite{KKM11}. 
   In particular, 
   Corollary 
   \ref{r_grows_faster_t_h_cor}
   implies
   the existence of a FA--presentation
   of a Cayley graph for which 
   the Cayley distance function   
   grows faster than any tower of exponents \cite[Remark~6.3]{eastwest19}.}. 
   
   These results lead to a 
   natural notion of 
   a compressibility rate $s(n)$
   of one FA--presentation 
   $\psi : L \rightarrow D$ 
   relative to another 
   $\psi_0 : L_0 \rightarrow D$
   for any given FA--presentable 
   structure $\mathcal{A}$ with the domain $D$.  
   The function $s(n)$ is defined
   as the maximal length of a shortest 
   normal form with respect to 
   $\psi_0$ for elements of 
   the domain $D$ having normal 
   forms of length less than 
   or equal to $n$ with respect
   to $\psi$, see Definition 
   \ref{comprate}. Then Theorem 
   \ref{r_grows_faster_than_T} 
   means that there exists 
   a FA--presentation of 
   $(\mathbb{N};\mathrm{S})$ for which 
   the compressibility rate 
   $s(n)$ relative to a unary 
   presentation of $(\mathbb{N};\mathrm{S})$ grows at least as fast as the function $T(n)$. We give more examples of 
   FA--presentable structures, including 
   the configuration spaces of one--tape Turing machines and Cayley graphs, for which
   there are FA--presentations $\psi_0$ 
   and $\psi$ such that the compressibility rate of $\psi$ relative to $\psi_0$  
   grows at least as fast as 
   the function $T(n)$. 
   However, for the Presburger arithmetic
   $(\mathbb{N};+)$ 
   we show that the compressibility 
   rate is always bounded from above 
   by a linear function, see Theorem
   \ref{incomp_thm}.

   The rest of the paper is organized as follows. In Section \ref{preliminaries_sec} we  
   briefly recall necessary definitions from  the field of FA--presentable structures. 
   In Section \ref{compr_nat_numbers_sec} 
   we discuss a numerical  characteristic $r(n)$  
   for FA--presentations of 
   $(\mathbb{N};\mathrm{S})$ and 
   construct a FA--presentation 
   for which $r(n)$ grows at least as 
   fast as the function $T(n)$. 
   In Section \ref{comp_rate_section}  
   we introduce a more general notion 
   of compressibility rate $s(n)$ for 
   FA--presentations of any 
   FA--presentable structure and 
   show that $s(n)$ is bounded from 
   above by a linear function for
   FA--presentations of 
   the Presburger arithmetic $(\mathbb{N};+)$.  
   In Sections 
   \ref{compr_config_tm_sec} and 
   \ref{compr_elem_of_groups} we
   show examples of 
   FA--presentations for the 
   configuration space of a Turing 
   machine and Cayley graphs of 
   some Cayley automatic groups for which 
   the compression rate grows at 
   least as fast as the function 
   $T(n)$.      
   Section \ref{conc_sec} concludes the 
   paper.
   
\section{Preliminaries} 
\label{preliminaries_sec}

  In this section we recall necessary 
  definitions and notations from 
  the field of FA--presentable structures. We assume that the reader is familiar with the basics of finite automata theory. 
    
  Let $\Sigma$ be an alphabet. For a given string 
  $w \in \Sigma^*$ we denote by $|w|$ the length of $w$.  
  We write $\Sigma_\diamond$  
  for the alphabet $\Sigma_\diamond = \Sigma \cup \{\diamond\}$, 
  where the padding symbol $\diamond$ is  not in $\Sigma$. 
  For a $k$--tuple of strings  
  $(w_1, \dots, w_k) \in \Sigma^{*k}$ the convolution 
  $w_1 \otimes \dots \otimes w_k \in \Sigma_\diamond^{k*}$ is 
  a string of length $|w| = \max\{|w_i|\, : \,i=1,\dots,k\}$ defined 
  as follows. For the $j$th symbol $(\sigma_1, \dots, \sigma_k)$ of $w$, 
  the symbol $\sigma_i$ for $i=1,\dots,k$
  is the $j$th symbol of $w_i$ if 
  $j \leqslant |w_i|$ and $\sigma_i = \diamond$, otherwise. 
 
  For a given relation $R \subseteq \Sigma^{*k}$ we denote 
  by $\otimes R$ the language 
  $\otimes R = \{w_1 \otimes \dots \otimes w_k \, : \, (w_1,\dots,w_k) \in R\}
  \subset \Sigma_\diamond^{k*}$.  
  The relation $R$  is called \emph{FA--recognizable}
  if the language 
  $\otimes R$ is regular.
  A FA--recognizable
  relation is also often referred to as 
  an \emph{automatic relation}.
   Alternatively, $R$ can be thought of 
  as a relation recognized by a synchronous $k$--tape finite automaton -- 
  a one--way Turing machine with $k$ input tapes.  
  
  For a $k$--ary function 
  $f: D^k \rightarrow D$ 
  we define 
  the $\mathrm{Graph}\,f $ to be the relation: 
  $$\mathrm{Graph}\,f = 
  \{(a_1,\dots,a_k,f(a_1,\dots,a_k))\, : \, (a_1,\dots,a_k) \in D^k\} \subseteq D^{k+1}.$$
  Similarly, we say that 
  a $k$--ary function 
  $f: D^k \rightarrow D$, where $D \subseteq \Sigma^{*}$, 
  is FA--recognizable   
  if the relation 
  $\mathrm{Graph}\, f$
  is FA--recognizable.
  A FA--recognizable
  	function is also often 
  referred to as an  
  \emph{automatic function}.

  A structure $\mathcal{A} = (D; R_1, \dots, R_\ell, f_1,\dots,f_m)$ 
  consists of a countable domain $D$, relations $R_1, \dots, R_\ell$ and 
  functions $f_1,\dots,f_m$ on $D$. 
  Let $\psi : L \rightarrow D$ be a surjective mapping 
  from  a language $L \subseteq \Sigma^*$ to the domain $D$. 
  For a given relation $R \subseteq D^n$ we denote 
  its preimage with respect to $\psi$ by $\widetilde{R}$: 
  $$
    \widetilde{R} = \{ (w_1,\dots,w_n) \in L^n \, : \, 
    (\psi(w_1),\dots, \psi(w_n)) \in R\}.
  $$ 
  We say that $\psi : L \rightarrow D$ is a FA--presentation 
  of the structure 
  $\mathcal{A}$ if $L$ is a regular language and 
  the relations 
  $\widetilde{R}_1,\dots,\widetilde{R}_\ell$ and 
  $\widetilde{\mathrm{Graph} \, f_1},\dots,\widetilde{\mathrm{Graph}\, f_m}$
  are FA--recognizable and 
  the equality relation $\{(u,v) \in L^2 \, : \, 
  \psi(u)=\psi(v)\}$ is FA--recognizable.
  We say that the structure $\mathcal{A}$ is FA--presentable if 
  it admits a FA--presentation.  FA--presentable 
  structures, for example, include $(\mathbb{N};+)$, $(\mathbb{Z}^n;+)$, the configuration spaces of Turing machines 
  and Cayley graphs of Cayley automatic 
  groups\footnote{Recall that a 
  finitely generated group $G$
  is called Cayley automatic if 
  its Cayley graph $\Gamma(G,S)$  
  for some finite set of generators 
  $S \subset G$ is a FA--presentable structure. 
  A FA--presentation
  $\psi : L \rightarrow G$ 
  of the Cayley graph $\Gamma(G,S)$ 
  is called a Cayley automatic representation of the group $G$.      
  Cayley automatic groups \cite{KKM11} naturally extend the class of automatic groups \cite{Epsteinbook} studied in geometric group theory.}. 
      
\section{Compressing Natural Numbers}
\label{compr_nat_numbers_sec}
 
 In this section we introduce a 
 numerical characteristic $r(n)$ 
 for FA--presentations of the structure
 $(\mathbb{N};\mathrm{S})$. We first
 show that $r(n)$ is bounded from above 
 by an exponential function for each
 FA--presentation of the Presburger 
 arithmetic $(\mathbb{N};+)$. Then 
 we construct a FA--presentation 
 of the structure $(\mathbb{N};\mathrm{S})$ 
 for which $r(n)$ 
 grows faster than any tower of 
 exponents of a fixed height.

 We denote by $\mathbb{N}$ the set of natural numbers which 
 includes zero and by $\mathrm{S}$ a successor 
 function defined on $\mathbb{N}$ by the identity $\mathrm{S}(x) = x+1$.  
 Let $L \subseteq \Sigma^*$ be a language and   
 $\psi : L \rightarrow \mathbb{N}$ be
 a FA--presentation 
 of the structure $(\mathbb{N};\mathrm{S})$.  
 For a given integer $n \geqslant 0$ we 
 define $L^{\leqslant n}$ to be the set of strings of 
 the language $L$ of length less than or equal to $n$: 
 $L^{\leqslant n} = \{ w \in  L \, : \, |w| \leqslant n \}$.

  \begin{definition} 
  \label{def_r_function} 	 
  	 For a given FA--presentation 
  	 $\psi: L \rightarrow \mathbb{N}$ of 
  	 the structure $(\mathbb{N};\mathrm{S})$ 
  	 we denote by $r$ the function 
  	 $r: \mathbb{N} \rightarrow \mathbb{N}$ defined 
  	 by the identities 
  	 $r (n) = \max \{ \psi(w) \, : \,w \in 
  	 L ^{\leqslant n} \}$ if   
  	 $L ^ {\leqslant n} \neq \varnothing$ and 
  	 $r(n)=0$ if $L^{\leqslant n} = \varnothing$.   	   
  \end{definition}	
   The function $r(n)$ is a numerical characteristic 
   of a FA--presentation 
   $\psi : L \rightarrow \mathbb{N}$ showing how large
   the number $\psi(u) \in \mathbb{N}$ can be 
   for a string $u \in L$ of length at most $n$.  
   For given nondecreasing functions 
   $r: \mathbb{N} \rightarrow \mathbb{N}$ and 
   $s: \mathbb{N} \rightarrow \mathbb{N}$   
   we say that $s \geqslant r$  (a function $s$ is greater than or equal to a 
   function $r$) if there exists an integer 
   $N$ for which $s(n) \geqslant r(n)$ 
   for all $n \geqslant N$.   
   The following proposition shows that if  
   $\psi: L \rightarrow \mathbb{N}$ is 
   a FA--presentation of the structure 
   $(\mathbb{N};+)$, then
   $r$ is less than or equal to some exponential 
   function.  
   \begin{lemma}
   \label{N_plus_exp_growth}	
   	  Let $\psi: L \rightarrow \mathbb{N}$ be a
   	  FA--presentation of the structure $(\mathbb{N};+)$.  
   	  There exists a constant 
   	  $\sigma>0$ such that 
   	  the exponential function 
   	  $\sigma^n$ is greater than
   	  or equal to $r(n)$. 
   \end{lemma}	
   \begin{proof}   	  
      Without loss of generality we can 
      assume that 
      $\psi : L \rightarrow \mathbb{N}$ 
      is bijective. 
      Indeed, 
      let $L' = \{ u \in L \, : \,  
      \forall v \left[ \psi(u) = \psi(v) \implies u \leqslant_{llex} v \right] \}$ and 
      $\psi' : L' \rightarrow 
      \mathbb{N}$ be the restriction of 
      $\psi$ onto $L' \subseteq L$, where $\leqslant_{llex}$ is a length--lexicographic ordering. The mapping 
      $\psi': L' \rightarrow \mathbb{N}$ is 
      a bijective FA--presentation of 
      the structure 
      $(\mathbb{N}; +)$.     
      Furthermore, the function   
      $r'(n) = \max \{ \psi'(w') \, : \,w' \in  L'^{\leqslant n} \}$ is equal 
      to $r(n) = \max \{ \psi(w) \, : \,w \in 
      L ^{\leqslant n} \}$.

      Now we notice that there exists a constant $c > 0$ such that for every triple $u,v,w \in L$ for which 
      $\psi (u) + \psi (v)  = \psi (w)$ the inequality 
      $\max  \{|u|,|v|\} \leqslant |w| + c$ holds.   
      This can be shown  as follows.  
      Since the relation $R =\{(u,v,w) \in L^3\, : \,
      \psi(u) + \psi(v) = \psi(w)\} \subseteq 
      \Sigma_\diamond ^{3*}$ is $3$--tape 
      FA--recognizable, there exists a finite automaton 
      $\mathcal{M}$ recognizing the language 
      $\otimes R= \{ u \otimes v \otimes w \, | \,
       \psi (u) + \psi (v)  = \psi (w) \}$.  
      Let $c$ be the number of states 
      in $\mathcal{M}$. 
      If $\max \{|u|, |v|\} > |w|+c$, then by the same argument as in the pumping lemma there exist
      $x,y,z \in \Sigma_\diamond ^{3*}$ for which 
      $u \otimes v \otimes w = xyz$,  $|x| \geqslant |w|$ and 
      $|y| \leqslant c$    
      such that every string $x y^n z$, $n \geqslant 0$
      is in the language $\otimes R$. This implies that there are 
      infinitely many $u',v' \in L$ for which  
      $\psi(u') + \psi (v') = \psi (w)$. 
      As $\psi : L \rightarrow \mathbb{N}$ is bijective, 
      we immediately get a contradiction. Therefore, 
      $\max \{|u|, |v|\} \leqslant |w|+c$.

   	  Let $m = \psi(w)$ and 
   	  $k = |w|+c$. There 
   	  exist exactly $m+1$ pairs
   	  $u,v \in L$ for which 
   	  $\psi(u) + \psi (v) = m$
   	  obtained from the $m+1$ 
   	  identities: $0+m=m,1+(m-1)=m,
   	  \dots,m+0=m$.   
   	  On the other hand, 
   	  the number of such pairs
   	  is bounded from above 
   	  by $1+\mu + \dots +\mu^{k} \leqslant 
   	  \frac{\mu^{k+1}-1}
   	  {\mu-1} \leqslant \mu^{k+1}$, where 
   	  $\mu = \# \Sigma$ is the number of 
   	  symbols in the alphabet $\Sigma$. It is assumed 
   	  that $\mu > 1$ as there 
   	  exists  
   	  no FA--presentation of
   	  the structure 
   	  $(\mathbb{N};+)$
   	  over a unary alphabet
   	  (this  can be  proved using the pumping lemma). 
   	  Therefore, 
   	  $m \leqslant m+1 \leqslant 
   	   \mu^{k+1}$ which 
   	  implies that  
   	$\psi(w) \leqslant 
   	 \mu^{c+1}\mu^{|w|}$.
   	 Therefore, for every 
   	 $w \in L ^{\leqslant n}$ 
   	 we have: 
   	 $\psi(w) \leqslant 
   	 \mu^{c+1} 
   	 \mu^{n}$.
   	 This implies that for any 
   	 $\sigma > \mu$, the 
   	 function $\sigma^n$ is greater than 
   	 or equal to $r(n)$.      	     	      	  
   \end{proof} 

\begin{remark}
  We note that the proof of Lemma \ref{N_plus_exp_growth} 
  cannot be generalized for the structure $(\mathbb{Z};+)$ as 
  for every $m \in \mathbb{Z}$
  there exist infinitely many $m_1,m_2 \in \mathbb{Z}$ 
  for which $m_1 + m_2 = m$. 
  Recall that the problem whether there exists a FA--presentation 
  of $(\mathbb{Z};+)$, for which the set
  of all nonnegative integers $\{z \in \mathbb{Z}\, : \,z \geqslant 0\}$
  is not regular, is open, 
  see \cite{SemiautomaticStrCSR,SemiautomaticStrJournal}. 	
  For an example of a FA--presentation 
  of $(\mathbb{Z};\mathrm{S})$ for which 
  the set of all nonnegative integers $\{z \in \mathbb{Z}\, : \,z \geqslant 0\}$ is not 
  regular see \cite{khoussainov2004definability}. 
  So the question whether the function 
  $\widetilde{r}(n)$, defined as  
  $\widetilde{r}(n) = \max \{ |\psi(w)| \, : \,w \in 
  L ^{\leqslant n} \}$ if 
  $L^{\leqslant n} \neq \varnothing$ 
  and $\widetilde{r}(n) = 0$ if 
  $L^{\leqslant n} = \varnothing$, 
  is bounded from above by an 
  exponential function for each FA--presentation 
  $\psi : L \rightarrow \mathbb{Z}$ of  $(\mathbb{Z};+)$ 
  cannot be trivially reduced to Lemma \ref{N_plus_exp_growth}.  
\end{remark}	
 Below we show that Lemma \ref{N_plus_exp_growth} fails  
 to hold for some FA--presentations of the structure $(\mathbb{N};\mathrm{S})$
 by constructing a concrete example for which the function 
 $r(n)$ grows faster than any tower of exponents of an arbitrary height, 
 see Corollary \ref{r_grows_faster_t_h_cor}.   
  
   Let $V$ be a set of all tuples 
   $v = (a,b,c,d)$ for which 
   $a,b,c$ and $d$ are integers such that
   the following three conditions are satisfied:
   \begin{enumerate}[I)] 
     \item{$a\geqslant 0,b \geqslant 0$, 
     	   $c \in \{ 2^k\, : \, k\geqslant 0 \}$ and 
     	   $d \in \{ 0, 1\}$;}
     \item{if $a>0$ and $b=0$, then $c>1$;}
     \item{if $a=0$, then $c=1$.}
   \end{enumerate}
   We define the function 
   $f : V \rightarrow V$ 
   according to the following six rules:    
   \begin{enumerate}[1)] 	
   	  \item{if $d=0, a>0$ and $b>0$, then 
   	        $f:(a,b,c,0) \mapsto (a,b-1,2c,0)$;}
      \item{if $d=0, a>0$ and $b=0$, then 
            $f:(a,0,c,0) \mapsto (a-1,c,1,0)$;}           
   	  \item{if $d=0, a=0$, then 
   	  	    $f : (0,b,1,0) \mapsto (0,b+1,1,1)$;}
      \item{if $d=1, c>1$, then 
            $f : (a,b,c,1) \mapsto
                 (a,b+1, \frac{c}{2},1)$;}        	    
      \item{if $d=1, c=1$ and $b \in 
      	       \{2^k \, : \, k > 0\}$, 
            then $f : (a,b,1,1) \mapsto (a+1,0,b,1)$;} 
      \item{if $d=1, c=1$ and $b \notin	
      	                 \{2^k \, : \, k > 0\}$, then 
                       $f: (a,b,1,1) \mapsto (a,b,1,0)$.}   	      	  
   \end{enumerate}	  
   For example, let us consecutively apply the function $f$  twenty three times to the tuple $(0,0,1,1)$. 
   We obtain: 
   \begin{equation*}
    \begin{split} 
   	 (0,0,1,1) \xrightarrow{6} (0,0,1,0) \xrightarrow{3} (0,1,1,1)
   	 \xrightarrow{6} (0,1,1,0) \xrightarrow{3} 
   	 (0,2,1,1) \xrightarrow{5} (1,0,2,1) 
   	 \xrightarrow{4}  \\ (1,1,1,1) \xrightarrow{6} 
   	 (1,1,1,0) \xrightarrow{1} (1,0,2,0)
   	 \xrightarrow{2} (0,2,1,0) \xrightarrow{3} 
   	 (0,3,1,1) \xrightarrow{6} (0,3,1,0) 
   	 \xrightarrow{3} \\  (0,4,1,1)  \xrightarrow{5} 
   	 (1,0,4,1) \xrightarrow{4} (1,1,2,1) 
   	 \xrightarrow{4} (1,2,1,1) \xrightarrow{5} (2,0,2,1)
   	 \xrightarrow{4} (2,1,1,1) \xrightarrow{6}  \\ 
   	 (2,1,1,0) 
   	 \xrightarrow{1} (2,0,2,0) \xrightarrow{2} (1,2,1,0) \xrightarrow{1} (1,1,2,0) \xrightarrow{1} (1,0,4,0) \xrightarrow{2} (0,4,1,0),	   
   	\end{split}  	   
   \end{equation*}	
  where each of the numbers above the arrows indicates one of 
  the six rules defining the function $f$. 
  Note that the tuple $(0,0,1,1)$ does not have a
  preimage with respect to $f$. 
  \begin{proposition}
  \label{f_is_correct}		
   The function $f: V \rightarrow V$ is correctly defined.     
  \end{proposition}
  \begin{proof} 
     In order to verify that $f : V \rightarrow V$ is correctly defined 
     one needs to check that for each of the six rules: if 
     $v \in V$, then $f(v) \in V$. That is, if the conditions 
     \text{I}, \text{II} and \text{III} hold for the tuple $v$, then they hold for
     the tuple $f(v)$ as well.   
     Clearly, the condition \text{I} holds for all $f(v)$, 
     $v \in V$.  
     
     Let us check it for the condition \text{II}. For the rule 1 
     we have $f(v)=(a,b-1,2c,0)$, so $2c>1$; therefore, 
     the conclusion of the condition \text{II} holds for $f(v)$.  
     For the rule 2 we have $f(v)=(a-1,c,1,0)$ for
     $c > 0$, so the assumption of the condition \text{II} is not valid for $f(v)$. For the rule 3 we have $f(v)=(0,b+1,1,1)$, so
     the assumption of the condition \text{II} is not valid 
     for $f(v)$
     as $b+1>0$. For the rule 4 we have $f(v) = (a,b+1,\frac{c}{2},1)$, 
     so the assumption of the condition \text{II} is not valid 
     for $f(v)$ as $b+1>0$. For the rule 5 we have 
     $f(v) = (a+1,0,b,1)$ for $b \in \{2^k \, : \, k>0\}$,
     so $b>1$; therefore, the conclusion of the 
     condition \text{II} holds for $f(v)$. 
     For the rule 6 we have $f(v) = (a,b,1,0)$ for 
     $b \notin \{2^k \, : \, k>0\}$.   
     If $a>0$ and $b=0$, then $v = (a,b,1,1)$ cannot be 
     in $V$ as the condition \text{II} is not satisfied for $v$.  
    
     Now let us check it for the condition \text{III}. For the rule 
     1 we have $f(v)=(a,b-1,2c,0)$ for $a>0$, so 
     the assumption of the condition \text{III} is not valid
     for $f(v)$. 
     For the rule 2 we have 
     $f(v)=(a-1,c,1,0)$, so the conclusion of the 
     condition \text{III} holds for $f(v)$. For the rule 3 we have 
     $f(v) = (0,b+1,1,1)$, so the conclusion of 
     the condition \text{III} holds for $f(v)$. For the rule 4 we have 
     $f(v)= (a,b+1,\frac{c}{2},1)$. If $\frac{c}{2}>1$, then $c>1$. Therefore, if $a=0$, then 
     $v=(a,b,c,1)$ cannot be in $V$ as the condition 
     \text{III} is not satisfied for $v$. For the rule 5 
     we have $f(v) = (a+1,0,b,1)$, so the
     assumption of the condition \text{III} is not valid for 
     $f(v)$ as $a+1 > 0$. For the rule 6 we have 
     $f(v) = (a,b,1,0)$, so the conclusion of the 
     condition \text{III} is valid for $f(v)$.                
  \end{proof}	
  
  \begin{proposition}
  \label{f_is_one-to-one}
  	 The function $f : V \rightarrow V$ is  one--to--one.     
  \end{proposition}  
  \begin{proof}
     In order to verify that $f:V \rightarrow V$ is 
     a one--to--one correspondence one needs to check that for 
     each pair of rules 
     $i$ and $j$, where 
     $i,j=1,\dots,6$, for all $u \in V$ and 
     $v \in V$ for which   
     the $i$th and $j$th rules are applied to 
     $u$ and $v$, respectively, if $f(u)=f(v)$, 
     then $u=v$. Clearly, this holds if $i=j$. 
     Also, if $i$ and $j$ belong to the different 
     sets of rules $\{1,2,6\}$ and $\{3,4,5\}$, 
     then $f(u) \neq f(v)$ because the fourth 
     components of $f(u)$ and $f(v)$ are different.  
     
     Let $i,j \in \{1,2,6\}$. If $i=1$ and $j=2$ or 
     $j=6$, for $u = (a,b,c,0)$ the equation 
     $f(u) = f(v)$ implies that
     $2c = 1$ which is impossible.
     If $i=2$ and $j=6$,  for $u = (a_1,0,c_1,0)$ and
     $v=(a_2,b_2,1,1)$ the equation 
     $f(u)=f(v)$ implies that 
     $c_1 = b_2$. By the condition \text{II} we have that 
     $c_1 >1$, so $c_1 \in \{2^k \, : \,k > 0\}$. 
     However, $b_2 \notin \{2^k \, : \, k>0\}$, so the 
     equation $c_1 = b_2$ is impossible.    
     
     Let $i,j \in \{3,4,5\}$. 
     If $i=5$ and $j=3$ or $j=4$, the equation $f(u)=f(v)$
     is impossible because the second component of $f(u)$ is 
     equal to $0$ while the second component
     of $f(v)$ is equal to $b+1>0$ in both cases.    
     If $i=3$ and $j=4$, 
     for $v = (a,b,c,1)$ the equation $f(u)=f(v)$
     implies that $a=0$. By the condition \text{III} we have that
     $c=1$. However, in the assumption of the rule 4 
     we have that $c>1$.       
  \end{proof}

  For a given integer $h \geqslant 0$ 
  we define $T(h)$ recursively by the formula  
  $T(h+1)=2^{T(h)}$ and the initial condition 
  $T(0)=1$. Let $\mathcal{T}$ be a set of towers of exponents 
  $\mathcal{T} = \{ T(h) \, : \, h \geqslant 0\}$; 
  that is, $\mathcal{T} = \{1,2,4,16,\dots,2^{2^{\dots^2}}, \dots \}$.  
  \begin{lemma} 
  \label{lemma_2m_notin_T}	 
  	 For each tuple of the form 
  	 $v= (0,2^m,1,1)$ for which 
  	 $2^m \notin \mathcal{T}$  there is an integer 
  	 $n \geqslant 0$ for which 
  	 $f^n (v) = (0,2^{m+1},1,1)$.
  \end{lemma}	
  \begin{proof} 
      Since $m > 1$ (otherwise $2^m \in \mathcal{T}$), we have that $2^m \in \{2^k\, : \,k>0\}$.  	
  	  Applying the rule 5 to $v$ we obtain that
  	  $f(v) = (1,0,2^m,1)$. Applying repeatedly the 
  	  rule 4 to $(1,0,2^m,1)$ we obtain 
  	  the tuple $(1, m, 1,1)$. 
  	  If $m \in \{2^k\, : \,k>0\}$, we continue 
  	  applying the rules 5 and 4 
  	  to obtain $(2,\log_2 m,1,1)$. 
  	  Continuing this process one gets  
  	  a tuple $(\ell+1,r,1,1)$, 
  	  where $\ell \geqslant 0$ and 
  	  $r =\log_2 ( \dots (\log_2 m)\dots) \notin 
  	  \{2^k\, : \,k>0\}$ 
  	  is obtained recursively from $m$ by applying
  	  the operator $\log_2$ exactly $\ell$ times. 
  	  Moreover, it follows from $2^m \notin \mathcal{T}$  that 
  	  $r>1$. We have 
  	  $(\ell+1,r,1,1) \xrightarrow{6} 
  	   (\ell+1,r,1,0)$.
  	  Applying repeatedly the rules 1 and 2 
  	  to the tuple $(\ell+1,r,1,0)$ we obtain 
  	  the tuple $(0,2^m,1,0)$. 
  	  Then we have that 
  	  $(0,2^m,1,0) \xrightarrow{3} (0,2^m+1,1,1)$.
  	  Applying repeatedly the rules 6 and 3 to  
  	  $(0,2^m+1,1,1)$ one finally gets the tuple $(0,2^{m+1},1,1)$. 
  \end{proof}	
   
  \begin{lemma} 
  \label{lemma_2m_in_T}	 
  	 For each tuple of the form $v=(0,2^m,1,1)$ for 
  	 which $2^m \in \mathcal{T}$ there exists an integer $n \geqslant 0$ 
  	 for which $f^n (v) = (a,1,1,0)$, where 
  	 $a$ is defined by the equation $T(a)=2^m$. 
  \end{lemma} 
  \begin{proof}
  	 For the case $2^m = 1$ we have: 
  	 $(0,1,1,1)
  	 \xrightarrow{6} (0,1,1,0)$. Now let $2^m > 1$. 
  	 Since  $2^m \in \mathcal{T}$, $m = T(\ell)$ for some 
  	 $\ell \geqslant 0$.  
  	 Applying repeatedly the rules 5 and 4 to $v$
  	 one gets a tuple $(\ell+1,1,1,1)$.  
  	 Finally we have 
  	 $(\ell+1,1,1,1) \xrightarrow{6} (\ell+1,1,1,0)$.
  	 The identity $m = T(\ell)$ implies that 
  	 $2^m = T(\ell+1)$.   
  \end{proof}	
 
  \begin{lemma} 
  \label{from_m_to_m_plus_1}	
  	For every integer $m \geqslant 0$ there exists 
  	an integer $n >0$ for which $f^n (m,1,1,0) = (m+1,1,1,0)$.   	
  \end{lemma}	
  \begin{proof}
     If $m=0$, we have: 
     $(0,1,1,0) \xrightarrow{3} 
     (0,2,1,1) \xrightarrow{5} (1,0,2,1) 
     \xrightarrow{4}  (1,1,1,1) \xrightarrow{6} 
     (1,1,1,0)$. Let $m>0$. Applying repeatedly the 
     rules 1 and 2 one gets a tuple 
     $(0,T(m),1,0)$. Applying to this tuple
     the rule 3 
     one gets a tuple $(0,T(m)+1,1,1)$. 
     The rules 6 and 3 should be then applied repeatedly
     to obtain a tuple $(0,2^{T(m-1)+1},1,1)$. 
     By Lemma \ref{lemma_2m_notin_T} applying 
     repeatedly the function $f$ to the tuple
     $(0,2^{T(m-1)+1},1,1)$ one  
     gets a tuple $(0,2^{T(m)},1,1)$. 
     Since $2^{T(m)} = T(m+1) \in  \mathcal{T}$, by Lemma   
     \ref{lemma_2m_in_T} applying 
     repeatedly the function $f$ to the tuple
     $(0,2^{T(m)},1,1)$ one finally gets a tuple
     $(m+1,1,1,0)$. 
  \end{proof}
 
  \begin{lemma} 
  \label{from_abcd_to_02m11}	
  	For every $v = (a,b,c,d) \in V$ there exist 
  	integers $m,n \geqslant 0$ for which 
  	$f^n(v) = (0,2^m,1,1)$. 
  \end{lemma}	
  \begin{proof}  
  	Let us consider first the case when 
  	$d=0$. If $a>0$, then applying repeatedly
  	the rules 1 and 2 to $v$ one gets a tuple 
  	of the form  $(0,b',1,0)$. Therefore, 
  	it is enough to analyze the case 
  	when $a=0$. If $b=0$ in a tuple 
  	$(0,b,1,0)$,  we have 
  	$(0,0,1,0) \xrightarrow{3} (0,1,1,1)$. 
  	Therefore, we can assume 
  	that $b>0$. Applying to $v$ the rule 
  	3 one gets a tuple $(0,b+1,1,1)$. 
  	If $b +1 = 2^k$ for some $k>0$, then we are done. 
  	Otherwise, the rules 6 and 3 should be  
  	repeatedly applied until one gets
  	a tuple $(0,2^k,1,1)$ for some $k>0$. 
  	
  	Now let us assume that $d=1$. 
  	If $c>1$, then applying repeatedly the rule 4
  	one gets a tuple $(a,b',1,1)$. 
  	Therefore, it is enough to analyse the 
  	case $c=1$. If $b \notin \{2^k\, : \, k>0 \}$, 
  	then applying to $v$ the rule 6 one gets a
  	tuple $(a,b,1,0)$. The lemma is  
  	already proved for the case $d=0$.  
  	If $b \in \{2^k\, : \, k>0 \}$, the rules 
  	5 and 4 should be repeatedly applied until 
  	one gets a tuple $(a',b',1,1)$ for 
  	$b' \notin \{2^k\, : \, k>0 \}$ -- the case
  	that we already analyzed.
  \end{proof}

  \begin{theorem} 
  \label{isomorph_thm}	 
  	 The structure $(V;f)$ is isomorphic to 
  	 $(\mathbb{N}; \mathrm{S})$. 
  \end{theorem}	
  \begin{proof} 
  	It follows from Proposition \ref{f_is_one-to-one} 
  	that $V$ can be decomposed into disjoint 
  	components $V_i \subseteq V, i \in  I$ for which 
  	$\bigcup\limits_{i \in I} V_i = V$, 
  	$f(V_i) \subseteq V_i$ and 
  	each structure $(V_i, f|_{V_i})$ is isomorphic 
  	to either $(\mathbb{N};\mathrm{S})$, 
  	$(\mathbb{Z};\mathrm{S})$ or 
  	$(\mathbb{Z}_n; \mathrm{S})$, where 
  	for a cyclic group $\mathbb{Z}_n$ the successor function
  	is given by $\mathrm{S}(x) = x+1 \mod n$ for $x \in \mathbb{Z}_n$. 
  	Suppose that there exist at least two disjoint 
  	components $V_i$ and $V_j$. 
  	It follows directly from Lemmas \ref{lemma_2m_notin_T}--\ref{from_abcd_to_02m11} 
  	that for every $u \in V_i$ and 
  	$v \in V_j$ there exist integers $r,s$ and $m$  
  	such that $f^r(u) = f^s(v) = (m,1,1,0)$. Since 
  	$f^r(u) \in V_i$ and $f^s(v) \in V_j$ we obtain 
  	that $V_i \cap V_j \neq \varnothing $, so we get a contradiction. 
  	Therefore, there is only one component. So
  	$(V;f)$ is either isomorphic to 
  	$(\mathbb{N}; \mathrm{S})$ or
  	$(\mathbb{Z}; \mathrm{S})$ as $V$ is infinite. Because 
  	$(0,0,1,1)$ does not have a preimage with 
  	respect to $f$, $(V;f)$ must be isomorphic 
  	to $(\mathbb{N}; \mathrm{S})$. 
  \end{proof}
  
  For a given nonnegative integer $n$ let 
  $n = \sum_{i=0}^k \beta_i 2^i$ be its binary decomposition, where $\beta_i \in \{0,1\}$ for 
  $i=0,\dots, k-1$ and $\beta_k=1$. We denote by 
  $\overline{n}$ the string $\beta_0\beta_1\dots \beta_k$, i.e., the standard binary representation of $n$ written in the reverse order. Similarly, for a given $4$--tuple of
  nonnegative integers $v = (a,b,c,d)$ we denote 
  by $\overline{v}$ the convolution of strings  
  $\overline{a} \otimes \overline{b} \otimes  
   \overline{c} \otimes \overline{d}$.    
  Let $L$ be the language of strings $\overline{v}$ 
  representing all $4$--tuples $v \in V$: 
  $L = \{ \overline{v} \, : \, v \in V\}$.
  We denote by $\varphi : L \rightarrow V$ 
  a bijection which for every $v \in V$ sends the string 
  $\overline{v} \in L$ to $v$.         
  \begin{proposition} 
  \label{FA_pres_Vf}   
     The map $\varphi : L \rightarrow V$ is a 
     FA--presentation of the structure $(V;f)$. 
  \end{proposition}	
  \begin{proof}
  	 To prove the proposition one needs to show that 
  	 $L$ is a regular language and the function 
  	 $f_L =  \varphi^{-1} \circ f \circ \varphi$ is automatic.
  	 For the reverse binary representation 
  	 of nonnegative integers that we use, 
  	 the set $\{ 2^k \, : \,k \geqslant 0\}$ corresponds to the language $0^*1$ which is regular.     
     So it is easy to see that for the presentation given by 
  	 $\varphi : L \rightarrow V$ each of the conditions
  	 \text{I}, \text{II} and \text{III}
  	  defining the set $V$ can be verified by a finite automaton. As the class of regular languages is closed 
  	 under intersection, the language $L$ is regular. 	 
  	 Similarly, for each of the six rules defining $f$ the assumption can be verified by a finite automaton. 
  	 Moreover, for the presentation of an integer 
  	 $n \geqslant 0$ by $\overline{n}$ the functions: $n \mapsto n+1$, $n \mapsto n-1$, $n \mapsto 2 n$ and 
  	 $n \mapsto \frac{n}{2}$ are 
  	 FA--recognizable. 
  	 This implies that the function $f_L : L \rightarrow L$ is FA--recognizable.  	  	 	   
  \end{proof}
  
   By Theorem \ref{isomorph_thm} there is 
   an isomorphism of the structures $(V;f)$ and  
   $(\mathbb{N};\mathrm{S})$ mapping the tuple 
   $(0,0,1,1) \in V$ to $0 \in \mathbb{N}$. We 
   denote this isomorphism by 
   $\tau : V \rightarrow \mathbb{N}$.
   Let $\psi$ 
   be the composition $\psi = \tau  \circ \varphi$.
   By Proposition \ref{FA_pres_Vf} the bijection 
   $\psi: L \rightarrow \mathbb{N}$ is a 
   FA--presentation of the structure $(\mathbb{N};\mathrm{S})$.  
   Let $r : \mathbb{N} \rightarrow \mathbb{N}$ be 
   the function corresponding to $\psi : L \rightarrow \mathbb{N}$
   as it is given 
   in Definition \ref{def_r_function}:  
   $r (n) = \max \{ \psi(w) \, :  \,w \in 
   L ^{\leqslant n} \}$ if   
   $L ^ {\leqslant n} \neq \varnothing$ and 
   $r(n)=0$ if $L^{\leqslant n} = \varnothing$.  
   \begin{theorem} 
   \label{r_grows_faster_than_T}	
   	  The function $r(n)$  
   	  is greater than or equal to $T(n)$.   
   \end{theorem}	
   \begin{proof} 
             
      For a given $n>2$, let 
      $u_n \in \{0,1\}^*$ be the string  
      $u_n = \overline{2^{n-1}}$,
      that is, 
      $u_n = 0^{n-1} 1$. 
      Let $m =2^{n-1}$. 
      The string 
      $w_n = u_n \otimes 1 \otimes 1 \otimes 0$ represents a tuple 
      $(m,1,1,0)$: 
      $\varphi (w_n) =(m,1,1,0)$.  
      By Lemma \ref{from_m_to_m_plus_1}, 
      there exists $\ell>0$ 
      for which 
      $f^\ell (m,1,1,0) = (m+1,1,1,0)$.
      In particular, we have:  
      $(m,1,1,0) \xrightarrow{f} 
       \dots \xrightarrow{f} 
       (1,T(m-1),1,0) 
       \xrightarrow{f} \dots 
        \dots  
       \xrightarrow{f} (1,0,T(m),0)
       \xrightarrow{f} \dots 
       \xrightarrow{f} (m+1,1,1,0)$,
       where in the subsequence 
       $(1,T(m-1),1,0) 
       \xrightarrow{f} \dots 
       \dots  
       \xrightarrow{f} (1,0,T(m),0)$
       the function $f$ is applied 
       exactly $T(m-1)$ times. Therefore, 
       $\ell \geqslant T(m-1)$. 
       Now let $v_n = \overline{2^{n-1}+1} = 10^{n-2}1$ and 
       $w_n'=v_n \otimes 1 \otimes 1 
        \otimes 0$. 
       The string $w_n'$ 
       represents the tuple 
       $(m+1,1,1,0)$:
       $\varphi (w_n')= (m+1,1,1,0)$.
       Clearly, $|w_n|= |w_n'| = n$. 
       Therefore, $r(n) \geqslant
       \psi(w_n') \geqslant  
       \ell \geqslant T(m-1)$. So,  
       $r(n) \geqslant T(2^{n-1}-1)$
       for all $n>2$. 
       Since 
       $2^{n-1} -1 \geqslant n$ for 
       $n>2$, we have that: 
       $r(n) \geqslant T(n)$ for 
       all $n>2$ which implies that 
       $r$ is greater than or equal to $T$.          
   \end{proof}	
   For a given integer $h \geqslant 0$ let $t_h (n)$
   be the function defined recursively 
   by the formula $t_{h+1}(n) =2^{t_h(n)}$ and 
   the initial condition $t_0 (n) = n$; that is, 
   $t_1(n) = 2^n, t_2 (n) = 2^{2^n}, 
   t_3 (n) = 2^{2^{2^n}}$ and etc. 
   \begin{corollary}
   \label{r_grows_faster_t_h_cor}
      For each $h \geqslant 0$, $r \geqslant t_h$. That is, the function $r$ grows 
      faster than any tower of exponents of 
      a fixed height.    	
   \end{corollary}	
   \begin{proof} 
   	  This immediately follows from 
   	  Theorem \ref{r_grows_faster_than_T} and    	  
   	  a simple observation that the function $T$ is greater than or equal to $t_h$ for 
   	  every $h \geqslant 0$.
   \end{proof}

 \section{Compressibility Rate}
 \label{comp_rate_section}
 
 In this section we 
 extend the notion of  
 a numerical characteristic 
 $r(n)$ defined  for 
 FA--presentations of 
 $(\mathbb{N};\mathrm{S})$ to a more 
 general notion of a
 compressibility rate $s(n)$ of 
 one FA--presentation relative
 to another for any given 
 FA--presentable structure. 
 We show that for each 
 pair of FA--presentations of the 
 Presburger arithmetic 
 $(\mathbb{N};+)$ the compressibility rate is bounded
 from above by a linear function.

 Let  
 $\mathcal{A} = (D; R_1, \dots, R_\ell, f_1,\dots,f_m)$ be a FA--presentable structure
 and $\psi_0 : L_0 \rightarrow D$ be 
 a FA--presentation of $\mathcal{A}$. 
 Let $\psi: L \rightarrow D$ also be a 
 FA--presentation of $\mathcal{A}$. 
 We define  
 $\xi : L \rightarrow \mathbb{N}$ 
 to be a function which maps a given 
 string $w \in L$ to 
 $$\xi(w) = \min \{|v| \, : \,
  \psi_0 (v) = \psi(w), v \in L_0 \}.$$ 
 The value $\xi(w)$ for $w \in L$ is the minimal length of a representative of the element  
 $\psi(w) \in D$ with respect to
 the FA--presentation 
 $\psi_0 : L_0  \rightarrow D$.   
  
 \begin{definition}
 \label{comprate}
 	For a given FA--presentation 
 	$\psi : L \rightarrow D$
 	of the structure $\mathcal{A}$ 	
 	let $s: \mathbb{N} \rightarrow \mathbb{N}$
 	be a function defined as follows.  
    For a given 
 	$n \in \mathbb{N}$, 
 	if $L^{\leqslant n} = \varnothing$, then 
 	$s(n) = 0$ and, if  $L^{\leqslant n} \neq \varnothing$, then 
 	$s(n) = 
 	\max \{ \xi(w) \, : \, w \in  L^{\leqslant n}\}$.      
 \end{definition}	       
 
 For infinitely many 
 $n$ the quotient 
 $\frac{s(n)}{n}$ is a compression ratio achieved
 for some strings in $L_0$.    
 We will call the function $s(n)$ 
 {\it compressibility rate} of the FA--presentation 
 $\psi : L \rightarrow D$ relative to 
 the FA--presentation 
 $\psi_0 : L_0 \rightarrow D$.    
 
 Let $\Sigma_0 = \{0\}$ be a unary 
 alphabet and 
 $u_0: \Sigma_0 ^* \rightarrow \mathbb{N}$
 be a  unary 
 FA--presentation of the structure 
 $(\mathbb{N};\mathrm{S})$ which sends 
 a string over the alphabet $\Sigma_0$ 
 to its length.   
 Theorem \ref{r_grows_faster_than_T} implies 
 that there exists a FA--presentation 
 of the structure $(\mathbb{N};\mathrm{S})$ 
 for which the compressibility rate 
 relative to the FA--presentation 
 $u_0$ is greater than or
 equal to $T(n)$. In particular, it grows 
 faster than any tower of exponents of 
 a fixed height, see Corollary \ref{r_grows_faster_t_h_cor}.  
 In Sections \ref{compr_config_tm_sec} and
 \ref{compr_elem_of_groups} we provide more 
 examples of FA--presentable structures and 
 their FA--presentations for which compressibility 
 rate grows faster than any tower of exponents.   
 
 However, not  every 
 FA--presentation admits compression.  
 We will say that a FA--representation 
 $\psi_0 : L_0 \rightarrow D$ is 
 {\it incompressible}
 if for every  FA--presentation 
 $\psi: L \rightarrow D$ of 
 the structure $\mathcal{A}$ the 
 compressibility rate 
 $s(n)$ is bounded from above by a linear 
 function $cn$ for some constant $c>0$ 
 which depends on the FA--presentation 
 $\psi : L \rightarrow D$.
 \begin{theorem} 
 \label{incomp_thm}	
 	Every FA--presentation 
 	of the structure $(\mathbb{N};+)$ is incompressible. 
 \end{theorem}	
 \begin{proof}
    Let $\psi_0 : L_0 \rightarrow 
         \mathbb{N}$ and 
         $\psi: L \rightarrow \mathbb{N}$            
    be FA--presentations 
    of the structure $(\mathbb{N};+)$. 
    Similarly to the proof of Lemma
    \ref{N_plus_exp_growth}, without 
    loss of generality, we can assume that 
    both FA--presentations $\psi_0$ and 
    $\psi$ are bijective. Then the function 
    $s(n)$ can be defined in a more simple way:
    $s(n)= \max \{|v|\, : \,\psi_0 (v) = \psi(w), 
    w \in L ^{\leqslant n} \}$.
    
    Now we notice that there exist constants 
    $c_0,d_0 >0$ such that 
    the inequality $\psi_0 (v) \leqslant 2^n$ implies
    that $|v| \leqslant c_0 n + d_0$ for all 
    $n \in \mathbb{N}$.  
    To see this, let $v_k \in L_0, k = 0,1,2, 
    \dots$  be the representative 
    of $2^k$ with respect to $\psi_0$: 
    $\psi_0(v_k) = 2^k$. 
    Since the relation $R_0 =\{(u,v,w) \in L_0 ^3\, : \,
    \psi_0(u) + \psi_0 (v) = \psi_0 (w)\} \subseteq 
    \Sigma_\diamond ^{3*}$ is $3$--tape 
    FA--recognizable, the relation 
    $R_0 ' = \{ (u,w) \in L_0 ^2 \,|\, 2\psi_0(u) = \psi_0(w) \}$ 
    is $2$--tape FA--recognizable. 
    Therefore, there exists a finite automaton 
    $\mathcal{M}$ recognizing the language 
    $\otimes  R_0' = \{ u \otimes w \, | \, 
     2 \psi_0 (u) = \psi_0 (w) \}$.
    Let $c_0$ be the number of states in $\mathcal{M}$. 
    If $|v_{k+1}| - |v_k | > c_0$, then 
    by the same argument as in the pumping lemma 
    there exist $x,y ,z\in \Sigma_\diamond ^{2*}$ 
    for which $v_k \otimes v_{k+1} = xyz$, $|x| \geqslant |v_k|$ 
    and $|y| \leqslant c_0$  such that every string 
    $xy^nz$, $n \geqslant 0$ is in the 
    language $\otimes R_0'$. This implies that 
    there are infinitely many $v' \in L_0$ for which 
    $2\psi_0 (v_k) = \psi_0 (v')$. 
    Since $\psi_0 : L_0 \rightarrow \mathbb{N}$ is bijective, 
    we get a contradiction. Therefore, 
    $|v_{k+1}| - |v_k| \leqslant c_0$.      
    Let $d_0' = |v_0|$. Then we have 
    $|v_k| \leqslant c_0 k + d_0'$ for all $k$. 
    The relation $\leqslant$ is first--order
    definable in $(\mathbb{N};+)$, so it
    is FA--recognizable. 
    Again, by using the pumping lemma argument one can show that 
    if $\psi_0(v) \leqslant \psi_0 (u)$, then 
    $|v| \leqslant |u| + d_0''$ for some constant
    $d_0''$. Therefore, if 
    $\psi_0 (v) \leqslant 2^n$, then
    $|v| \leqslant c_0 n + d_0' + d_0'' = 
    c_0 n + d_0$, where $d_0 = d_0' + d_0''$.       
    
    By Lemma \ref{N_plus_exp_growth}, 
    there exists a constant 
    $\sigma > 0$ for which 
    the function 
    $r(n) = \max\{\psi(w)\, : \, 
     w \in  L^{\leqslant n} \}$ 
    is less than or equal to 
    $\sigma^n$: 
    $r(n) \leqslant \sigma^n$.   
    Therefore, $r(n) \leqslant \sigma^n \leqslant 2^{\lceil \log_2 \sigma \rceil n}$. This implies that  
    $s(n) \leqslant c_0 \lceil \log_2 \sigma \rceil n + d_0$. 
    Let $c = c_0 \lceil \log_2 \sigma \rceil +1$. Then 
    finally we have 
    $s(n) \leqslant cn$.        
 \end{proof}	 
    
 \section{Compressing Configurations of a One--Tape  Turing Machine} 
 \label{compr_config_tm_sec}
 
 In this section we consider 
 a FA--presentable structure 
 defined by the set of all 
 possible configurations of 
 a one--tape Turing machine. 
 A standard encoding of 
 these configurations 
 gives a FA--presentation
 of this structure. We will 
 show that there exists 
 another FA--presentation 
 of the same structure   (encoding of configurations of a Turing machine) 
 for which the compressibility rate $s(n)$ relative to the standard encoding is greater than or equal to $T(n)$.

 Let $\Gamma$ be a finite set of symbols of cardinality at least two 
 which contains a blank symbol $\sqcup$ and $Q$ be a finite set of states
 containing a distinguished symbol $q_0 \in Q$; it is assumed that 
 $\Gamma \cap Q = \varnothing$. A deterministic one--tape Turing machine 
 $M$ over the alphabet $\Gamma$ with a set of states $Q$ and the 
 initial state $q_0$ is defined by the set of  
 commands $P_M$.    
 A configuration (instantaneous description) of $M$ is
 a string $X_1 \dots X_{i-1}q X_i X_{i+1} \dots X_n$, where 
 $X_1 \dots X_n \in \Gamma^*$ is the content written on the tape and  
 $q \in Q$ with the head pointing at $X_i$. 
 This way to present configurations 
 is standard regardless whether the tape is infinite or 
 semi--infinite, see, e.g., \cite{HopcroftUllman}; in the latter 
 case $X_1$ is the content of the leftmost cell. 
 We denote by $\mathcal{C}_{\Gamma,Q} \subseteq 
 (\Gamma \cup Q)^*$ the language 
 of configurations for all possible Turing machines 
 over the alphabet $\Gamma$ and a set of states $Q$. Clearly, 
 the language $\mathcal{C}_{\Gamma,Q}$ is regular.  
 Furthermore, the relation
 \begin{equation*}
 	R_{M} = \{ (\alpha,\beta) \in \mathcal{C}_{\Gamma,Q} 
 	\times \mathcal{C}_{\Gamma,Q} \, : \, \mathrm{there\,\,exists\,\,a\,\,  	 command\,\,in\,\,} P_M \mathrm{\,\,transforming\,\,}\alpha 
 	\mathrm{\,\,to\,\,}\beta\} 
 \end{equation*}	 
 is FA--recognizable for every Turing machine 
 $M$ over the alphabet $\Gamma$ with the set of states $Q$ \cite{KhoussainovNerode95}. 
 So the structure 
 $(\mathcal{C}_{\Gamma,Q}; R_M)$ 
 is FA--presentable and the 
 identity map 
 $\psi_0 : \mathcal{C}_{\Gamma,Q} 
 \rightarrow 
 \mathcal{C}_{\Gamma,Q}$ is a 
 FA--presentation.

 We construct a new FA--presentation
 $\psi : L \rightarrow \mathcal{C}_{\Gamma,Q}$ of 
 the structure 
 $(\mathcal{C}_{\Gamma,Q}; R_M)$ 
 as follows. 
 Let $\gamma$ be a nonblank symbol from $\Gamma$.  
 Any configuration $\xi \in \mathcal{C}_{\Gamma,Q}$ 
 can be written as a concatenation $\xi = \gamma^k \mu$ of strings 
 $\gamma^k$ for some $k \in \mathbb{N}$ and $\mu$, where the first 
 symbol of $\mu$ is not $\gamma$.  
 Now let $u_k$ be the string representing $k$
 with respect to the FA--presentation  of 
 $(\mathbb{N};\mathrm{S})$ constructed in Section \ref{compr_nat_numbers_sec};
 it is assumed 
 that $\Gamma$ does not contain 
 any symbol from 
 the alphabet of this FA--presentation.     
 We encode the configuration $\xi$ by a string
 $w = u_k \mu$ which is the concatenation of strings $u_k$ and $\mu$. 
 Let $L$ be the collection of all such strings $w$ encoding 
 all possible configurations from $\mathcal{C}_{\Gamma,Q}$.  
 Clearly, a mapping 
 $\psi : L \rightarrow \mathcal{C}_{\Gamma,Q}$ 
 which sends a string $w$ to the configuration $\xi$ is a bijection. 
 Moreover, for this mapping $\psi$ and each Turing machine
 $M$ over the alphabet $\Gamma$ with the set of states $Q$ 
 the relation $\widetilde{R_M}$ is FA--recognizable. 
 
 \begin{theorem} 
 	\label{s_grows_faster_than_T}
 	The compressibility rate $s(n)$ of $\psi$ relative 
 	to $\psi_0$
 	is greater than or equal to $T(n)$. 
 \end{theorem}	
 \begin{proof} 
 	First we notice that 
 	since $\psi_0$ is
 	the identity map and 
 	$\psi$ is bijective, 
 	the compressibility rate $s(n)$ of $\psi$ relative 
 	to $\psi_0$
 	takes a form:
 	$s(n) = \max \{ |\psi(w)| \, : \,w \in L^{\leqslant n} \}$
 	if   
 	$L ^ {\leqslant n} \neq \varnothing$ 
 	and $s(n)=0$ if $L^{\leqslant n} = \varnothing$. 
 	Let $q \in Q$. 
 	For a given $k \geqslant 0$ we denote by $\xi_k$ the configuration  
 	$\xi_k = \gamma^kq \sqcup $. 
 	Let $w_k = u_k q \sqcup \in L$ be the string representing $\xi_k$ with 
 	respect to $\psi$, i.e., $\psi(w_k) = \xi_k$.
 	We have: $|\xi_k| = k+2$ and $|w_k|= |u_k|+2$.
 	Therefore, $s(n+2) \geqslant r(n)  + 2$ for all $n > 0$.    
 	Now we note that in Theorem \ref{r_grows_faster_than_T} 
 	we actually proved a stronger inequality: $r(n) \geqslant T(2^{n-1}-1)$ 
 	for all $n>2$. Therefore,
 	$r(n) + 2 \geqslant T(2^{n-1}-1) + 2 \geqslant T(n+2)$ for all 
 	$n >3$; the latter inequality follows from a simple observation 
 	that $2^{n-1}-1>n+2$ for all $n>3$. Thus, $s(n)  
 	\geqslant r(n-2) +2 \geqslant T(n)$ for all $n>5$.   
 \end{proof}	
 
 \begin{remark}
 	Each bijective 
 	FA--presentation 
 	$\psi : L \rightarrow \mathcal{C}_{\Gamma,Q}$ 
 	of the structure 
 	$(\mathcal{C}_{\Gamma,Q}; R_M)$
 	defines an encoding of 
 	configurations of a Turing machine $M$ by strings from the language $L$. 
 	Moreover, if for strings $u \in L$ and  $v \in L$ encoding configurations 
 	$\alpha = \psi (u)$ and $\beta  = \psi(v)$, respectively, there exists a command 
 	in $P_M$ transforming $\alpha$ to $\beta$, the string $v$ 
 	can be computed on some deterministic one--tape position--faithful Turing machine (see \cite{LMCS13} for the formal definition of a one--tape position--faithful Turing machine) 
 	from the input string $u$ in linear time. This is because 
 	being an automatic function is equivalent to being one  
 	computed on a deterministic one--tape position--faithful 
 	Turing machine in linear time \cite{LMCS13}. 	
 \end{remark}

\section{Compressing Elements in Cayley Automatic Groups}
\label{compr_elem_of_groups} 

In this section we consider 
FA--presentations of 
Cayley graphs for Cayley 
automatic groups. 
Such FA--presentations 
are referred to as Cayley automatic representations. 
The groups $G$ considered in this section are  
free abelian groups, free groups, Baumslag--Solitar 
groups and semidirect products.
We start with
fixing some known FA--presentations
$\psi_0: L_0 \rightarrow G$ of 
these groups. 
Then we construct new 
FA--presentations 
$\psi : L \rightarrow G$ 
for which the compressibility  
rate $s(n)$ relative to $\psi_0$ is greater than or equal to $T(n)$. 
That is, we show
the result analogous to 
Theorem \ref{s_grows_faster_than_T}. 
All FA--presentations that we consider in this section 
are bijective, so the 
compressibility rate takes
a form: $s(n) = \max \{|v| \, : \,
\psi_0 (v) = \psi(w), w \in L^{\leqslant n}\}$ 
if $L ^ {\leqslant n} \neq \varnothing$
and $s(n)=0$ if $L^{\leqslant n} = \varnothing$.

Throughout this section  
for a given integer $k \geqslant 0$  
we will denote by $u_k$ the string representing 
the integer $k$ with respect to the FA--presentation 
of $(\mathbb{N};\mathrm{S})$ constructed 
in Section \ref{compr_nat_numbers_sec}.

\subsection{Free Abelian Groups}

We first consider a natural
Cayley automatic
representation of the infinite cyclic  
group $Z = \langle a \rangle$ 
defined as follows.
Let $\Sigma_0 = \{a,a^{-1}\}$ and 
$L_0 = \{ a^k \, : \, k \in \mathbb{Z}\} 
\subseteq \Sigma_0^*$. 
We define $\psi_0 : L_0 \rightarrow Z$ 
to be a map sending  
a string $a^k \in L_0$ to the group 
element $a^k \in Z$. 
Now we define $L$ to be a language 
consisting of the strings
$a^{k}$ for  $k<0$ and $u_k$ for  
$k \geqslant 0$. 
Let $\psi : L \rightarrow Z$ be a map 
which sends a string $a^k \in L$ for $k<0$ 
and a string $u_k \in L$ for $k \geqslant 0$
to the group element $a^k \in Z$.      
It can be seen that 
the mapping $\psi : L \rightarrow Z$ 
is a Cayley automatic representation.  
Let $s(n)$ 
be the compressibility rate 
of 
$\psi : L \rightarrow Z$ 
relative to $\psi_0: L_0 \rightarrow Z$.  
The inequality $s(n) \geqslant T(n)$ 
immediately follows from  
Theorem \ref{r_grows_faster_than_T} 
as $\psi_0 (a^k) = \psi(u_k)$ and $|a^k|=k$
for all $k \geqslant 0$. 

Now let us consider a general case -- a free abelian group 
$Z^m = 
\langle a_1,\dots,a_m \, | \, [a_i,a_j]=e, 
i \neq j \rangle$. 
Let $\Sigma_0 = \{a_1, a_1^{-1},\dots,a_m,a_m ^{-1}\}$, 
$L_0 = \{a_1^{k_1} \dots a_m ^{k_m}\, : \,
k_i \in \mathbb{Z} , i =1,\dots, m \}$ 
and $\psi_0 : L_0 \rightarrow Z^m$ 
be a Cayley automatic representation of 
the group $Z^m$ sending a 
string $a_1^{k_1} \dots a_m ^{k_m} \in L_0$ 
to the group element 
$a_1^{k_1} \dots a_m ^{k_m} \in Z^m$. 
We define $L$ to be a language consisting 
of the strings 
$a_1^{k_1} a_2^{k_2} \dots a_m ^{k_m}$ for 
$k_1 <0$ and 
$u_{k_1} a_2 ^{k_2} \dots a_m ^{k_m}$ 
for $k_1 \geqslant 0$.
Let $\psi : L \rightarrow Z^m$ be 
a map which sends a string 
$a_1^{k_1} a_2 ^{k_2} \dots a_m ^{k_m} \in L$
for $k_1 <0$ and a string 
$u_{k_1} a_2^{k_2} \dots a_m^{k_m} \in L$ for 
$k_1 \geqslant 0$ to the group 
element $a_1^{k_1} a_2^{k_2} \dots a_m ^{k_m} \in Z^m$.
The mapping $\psi : L \rightarrow Z^m$ is a Cayley automatic representation. 
Similarly, 
$s(n) \geqslant T(n)$ for 
the compressibility rate 
of $\psi$ relative to $\psi_0$.

\subsection{Free Groups} 

In this part we consider 
a free group over $m$ generators   
$\mathbb{F}_m = 
\langle a_1, \dots, a_m \rangle$. 
Recall that $\mathbb{F}_m$ 
as a set consists of 
all reduced words over the alphabet
$\{a_1,a_1^{-1}, \dots, 
a_m,a_m^{-1}\}$.    
Let 
$\Sigma_0 = 
\{a_1,a_1^{-1}, \dots, 
a_m,a_m^{-1}\}$ and 
$L_0 \subseteq \Sigma_0 ^*$ be 
the language of reduced words 
over $\Sigma_0$.   
We define 
$\psi_0 : L_0 \rightarrow 
\mathbb{F}_m$ 
to be a Cayley automatic representation 
identifying a reduced word from $L_0$ 
with the corresponding element 
in $\mathbb{F}_m$.    
Each reduced word $w \in L_0$
can be written as a concatenation 
$w = a_1 ^{k} w'$, where 
$k \in \mathbb{Z}$, $w' \in L_0$ and 
the first symbol of $w'$ is not $a_1$ or $a_1^{-1}$. 
Now we define $L$ to be 
the language consisting of all concatenations 
$a_1^k w'$ for $k<0$ and $u_k w'$ for 
$k \geqslant 0$. 
Let $\psi : L \rightarrow \mathbb{F}_m$ 
be a map which sends a string 
$a_1^k w' \in L$ for $k<0$ and 
$u_k w' \in L$ for $k \geqslant 0$ 
to the group element  $w = a_1 ^{k} w' \in \mathbb{F}_m$. 
Clearly, $\psi : L \rightarrow \mathbb{F}_m$ 
is a Cayley automatic representation. 
As $\psi_0 (a_1 ^k) = \psi (u_k)$ and 
$|a_1 ^k| =k$ for all $k \geqslant 0$,
the inequality $s(n) \geqslant T(n)$ for the compressibility 
rate $s(n)$ of $\psi$ relative to 
$\psi_0$  is a 
straightforward corollary of Theorem 
\ref{r_grows_faster_than_T}. 

\subsection{Baumslag--Solitar Groups}

In this part we consider the family of Baumslag--Solitar 
groups  
$BS(p,q) = \langle a,t \, :  \, t a^p t^{-1}  = a^q \rangle$
for $1 \leqslant p < q$.  
Recall that each group element 
$g \in BS(p,q)$ can be uniquely 
written 
as a reduced word 
$w_\ell t^{\varepsilon_\ell} \dots
w_1 t^{\varepsilon_1} a^m$, 
where $\varepsilon_i \in \{+1,-1\}$,  
$w_i \in \{\varepsilon, a, \dots, a^{p-1}\}$  
if $\varepsilon_i  = -1$, 
$w_i \in  \{\varepsilon, a, \dots, a^{q-1}\}$ 
if $\varepsilon_i = +1$ and $m \in \mathbb{Z}$.
The reader can look up a  general  result about normal forms in HNN extensions of groups in, e.g, \cite{LyndonSchuppbook}. 
We represent an element 
$g = w_\ell t^{\varepsilon_\ell} \dots
w_1 t^{\varepsilon_1} a^m$ as a concatenation 
of a string 
$\widetilde{w}= w_\ell 
t^{\varepsilon_\ell} \dots
w_1 t^{\varepsilon_1}$ and a string $z$, which 
is a $q$--ary representation
of an integer $m$. 
Let $L_0$ be the language of 
all such concatenations 
$u = \widetilde{w}z$ and 
$\psi_0 : L_0 \rightarrow BS(p,q)$ 
be a bijection which sends 
a string $u  = \widetilde{w}z \in L_0$ to a group element 
$g = \widetilde{w} a^m \in BS(p,q)$. 
This bijection 
$\psi_0 : L_0 \rightarrow BS(p,q)$ is a Cayley 
automatic representation of the 
group $BS(p,q)$ \cite{dlt14}. 
We denote by $\tau$ the maximal prefix 
of the string $u$ which 
is of the form $\tau = t^k$ for 
$k \geqslant 0$. That is, 
$u = \tau \omega = t^k \omega$ 
and the first 
symbol of the suffix $\omega$ 
is not $t$.   
Now we define $L$ to be the language 
consisting of all concatenations  
$u_k \omega$.  
Let $\psi : L \rightarrow BS(p,q)$ 
be a map which sends a string 
$u_k \omega \in L$ to a group element
$g = \widetilde{w} a^m \in BS(p,q)$. 
It can verified that 
$\psi : L \rightarrow BS(p,q)$ 
is also a Cayley automatic representation. 
As $\psi_0 (t^k) = \psi(u_k)$ and 
$|t^k|=k$ for all $k \geqslant 0$, the inequality $s(n) \geqslant T(n)$ for the compressibility 
rate $s(n)$ of $\psi$ relative to 
$\psi_0$ follows from Theorem 
\ref{r_grows_faster_than_T}.

\subsection{Semidirect Products 
	$Z^2 \rtimes_A Z$}  

In this part we consider a 
family of semidirect products 
$Z^2 \rtimes_A Z$ 
for $A \in \mathrm{GL}(2,\mathbb{Z})$.  
Let us consider any  
FA--presentation 
$\psi' : L' \rightarrow Z^2$
of the structure $(Z^2;f_A)$, 
where 
$f_A : Z^2 \rightarrow 
Z^2$ is an 
automorphism  
mapping 
${z_1 \choose z_2} \in Z^2$ to 
$A {z_1 \choose z_2} \in Z^2$. 
Let $a$ be a generator of the 
subgroup $Z \leqslant 
Z^2 \rtimes_A Z$.
We denote by $L_0$ the language 
of all concatenations 
$a^k v$ for $k \in \mathbb{Z}$ and 
$v \in L'$; it is assumed that the 
alphabet of $L'$ does not contain  
the symbols $a$ and $a^{-1}$.   
Let $\psi_0 : L_0 \rightarrow 
Z^2 \rtimes_A Z$
be a bijection which sends a 
string $a^k v$ to the group 
element 
$\left(a^k,{z_1 \choose z_2} \right) \in 
Z^2 \rtimes_A Z$, 
where ${z_1 \choose z_2} = \psi'(v)$. 
This bijection  
$\psi_0 : L_0 \rightarrow 
Z^2 \rtimes_A Z$ 
is a Cayley automatic representation 
of the group 
$Z^2 \rtimes_A Z$;
see \cite{BK2020} where 
such 
Cayley automatic representations  
are used. We define $L$ to be a 
language consisting of all 
concatenations $a^k v$ for 
$k<0$ and $u_k v$ for $k \geqslant 0$. 
Let $\psi : L \rightarrow Z^2 \rtimes_A Z$ be a map 
which sends a string $a^kv \in L$ for 
$k <0$ and $u_k v \in L$ 
for $k \geqslant 0$ 
to the group element 
$\left(a^k,{z_1 \choose z_2} \right) 
\in Z^2 \rtimes_A Z$, 
where ${z_1 \choose z_2} = \psi'(v)$.  
Let us additionally assume that the empty string $\varepsilon \in L'$; if 
$\varepsilon \notin L'$ one can 
always change any element of $L'$ to 
the empty string $\varepsilon$ -- this 
will give a new FA--presentation 
of the structure $(Z^2;f_A)$.    
As $\psi_0 (a^k) = \psi(u_k)$ and $|a^k|=k$
for all $k \geqslant 0$, 
the inequality $s(n) \geqslant T(n)$ for the compressibility 
rate $s(n)$ of $\psi$ relative to 
$\psi_0$ follows from Theorem 
\ref{r_grows_faster_than_T}. 

   \section{Conclusion and Open Questions}  
  \label{conc_sec}
  The key result of this paper 
  is a construction of 
  a FA--presentation of the structure 
  $(\mathbb{N}; \mathrm{S})$ 
  such that for every $n \geqslant 0$
  there is a string of length at most $n$
  from the domain of this FA--presentation 
  which encodes an integer that is
  greater than or equal to $T(n)$, 
  where $T(n)$ is defined recursively 
  by the identities $T(n+1)=2^{T(n)}$ and $T(0)=1$. In particular,
  $T(n)$ grows faster than any tower of exponents of a fixed height. 
  This result naturally leads to the notion of a compressibility 
  rate defined for a pair  of FA--presentations for 
  any FA--presentable structure.  We show examples  
  when this compressibility rate grows at least as 
  fast as $T(n)$. We show that for FA--presentations 
  of the Presburger arithmetic 
  $(\mathbb{N};+)$ it is  
  bounded by a linear function.
  We leave the following questions for future consideration.  
  \begin{itemize} 
  	\item{Is it true that the compressibility rate 
  		for FA--presentations of the structure 
  		$(\mathbb{Z};+)$ is always bounded from 
  		above by a linear function?}
  	\item{Is it true that for every FA--presentation 
  		$\psi_0$ of $(\mathbb{N};\mathrm{S})$ there 
  		exists a FA--presentation $\psi$ for which the 
  		compressibility rate of $\psi$ relative 
  		to $\psi_0$ is  bounded from below 
  		by the function $T(n)$?} 
  	\item{The notion of a compressibility rate is 
  		valid for semiautomatic structures 
  		\cite{SemiautomaticStrCSR,SemiautomaticStrJournal}.
  		Is it true that for semiautomatic presentations
  		of the Presburger arithmetic 
  		$(\mathbb{N};+)$
  		the compressibility 
  		rate is bounded from above by a linear function?} 
  \end{itemize}

  \section*{Acknowledgement} 
  The authors would like to thank the anonymous reviewer 
  for comments. Dmitry 
  Berdinsky thanks Murray Elder for useful discussions.      
  The authors thank the Institute for Mathematical Sciences 
  at the National University of Singapore for support. 
  Sanjay Jain's research was partially supported by 
  the National University of Singapore grant E-252-00-0021-01. Bakhadyr Khoussainov's research was partially funded by the NSFC China, the grant number 62172077. 
  Sanjay Jain and Frank Stephan were partially 
  supported by the Singapore Ministry of
  Education Academic Research Fund Tier 2 grant
  MOE2019-T2-2-121 / R-146-000-304-112.

\bibliographystyle{plain}

\bibliography{tmconfig_bibliography}

\end{document}